%% file: main.tex
\documentclass{amsart}

\usepackage[hmargin=3cm,vmargin=3cm]{geometry}
\usepackage{ccicons}
\usepackage{amsmath,stmaryrd}
\usepackage{amssymb}
\usepackage{amsxtra}
\usepackage{mathtools}
\usepackage{enumerate}
\usepackage{comment}
\usepackage{graphicx}
\usepackage[czech,english]{babel}
\usepackage{mdwtab}
\usepackage[OT1]{fontenc}
\usepackage[all]{xy}
\usepackage{todonotes}
\usepackage{url}
\usepackage[ruled,vlined]{algorithm2e}

\usepackage{color}
\definecolor{shadecolor}{gray}{.85}%
\definecolor{tintedcolor}{gray}{.80}%
\usepackage{framed}%
\definecolor{mytintedcolor}{gray}{.95}%
\makeatletter
\newdimen\svparindent
\newcounter{tmpthm}
\newenvironment{mytinted}{%
  \MakeFramed {\FrameRestore}}%
{\endMakeFramed}
{\endlist\end{mytinted}\egroup}
\makeatother

\newtheorem{theorem}{Theorem}
\newtheorem{proposition}{Proposition}
\newtheorem{corollary}{Corollary}
\newtheorem{lemma}[theorem]{Lemma}
\newtheorem{claim}{Claim}
\newtheorem{conjecture}{Conjecture}

\theoremstyle{definition}
\newtheorem{definition}[theorem]{Definition}

\theoremstyle{remark}
\newtheorem{remark}[theorem]{Remark}

\newcommand{\ERCagreement}{\thanks{
\begin{minipage}{.67\textwidth}This paper is a part of projects that have received funding from the European Research Council (ERC) 
under the European Union's Horizon 2020 research and innovation programme (grant agreements No. 677651 -- {\sc{Total}}, and No. 810115 -- {\sc Dynasnet}). 
Xuding Zhu's research is supported by grant numbers NSFC 11971438, ZJNSF LD19A010001. The research of Jaroslav Ne\v set\v ril, Patrice Ossona de Mendez, and Xuding Zhu is partially supported by 111 project of the Ministry of Education of China.
\end{minipage}\hfill\begin{minipage}{.25\textwidth}\includegraphics[width=\textwidth]{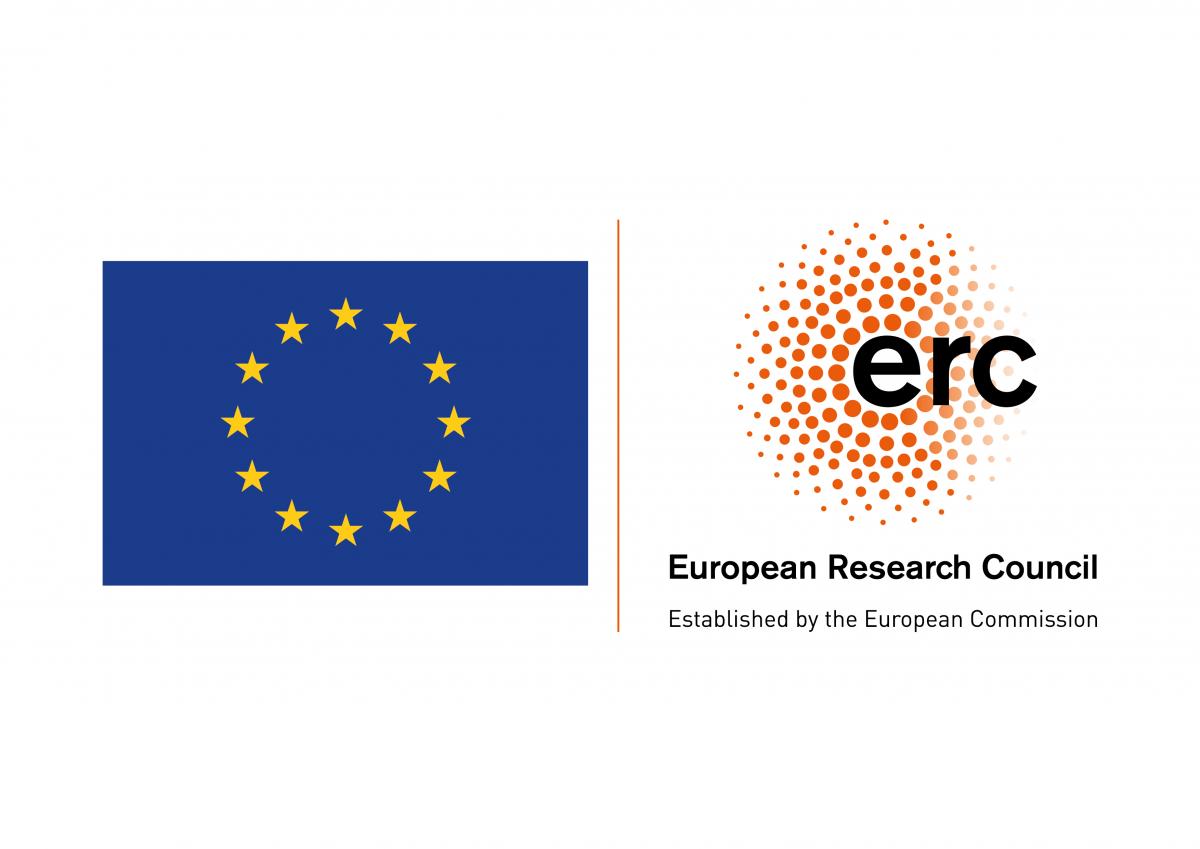}\end{minipage}\hfill}}

\newcommand{\sfrac}[2]{#1/#2}

\newcommand{\myparagraph}[1]{\medskip \noindent{\bf{#1}}}

\newcommand{\cls}[1]{\overline{#1}}
\newcommand{\WReach}{\mathrm{WReach}}
\newcommand{\grad}{\nabla}
\newcommand{\dist}{\mathrm{dist}}
\newcommand{\wcol}{\mathrm{wcol}}
\newcommand{\col}{\mathrm{col}}
\newcommand{\girth}{\mathrm{girth}}
\newcommand{\chisub}{\chi_\mathrm{sub}}
\newcommand{\N}{\mathbb{N}}
\newcommand{\Cc}{\mathcal{C}}
\newcommand{\Dd}{\mathcal{D}}
\newcommand{\Oh}{\mathcal{O}}
\newcommand{\eps}{\varepsilon}

\newcommand{\clust}{\mathcal{X}}

\renewcommand{\leq}{\leqslant}
\renewcommand{\geq}{\geqslant}

\def\cqedsymbol{\ifmmode$\lrcorner$\else{\unskip\nobreak\hfil
\penalty50\hskip1em\null\nobreak\hfil$\lrcorner$
\parfillskip=0pt\finalhyphendemerits=0\endgraf}\fi} 

\newcommand{\cqed}{\renewcommand{\qed}{\cqedsymbol}}

\begin{document}

\title{Clustering Powers of Sparse Graphs}
\ERCagreement
\author{Jaroslav Ne\v set\v ril}
\address[Jaroslav Ne\v set\v ril]{Institute for Theoretical Computer Science, Charles University, Prague, Czech Republic}
\email{nesetril@iuuk.mff.cuni.cz}
\author{Patrice Ossona~de~Mendez}
\address[Patrice Ossona~de~Mendez]{Centre d'Analyse et de Math\'ematiques Sociales (UMR 8557) and CNRS, Paris, France}
\email{pom@ehess.fr}
\author{Micha{\l} Pilipczuk}
\address[Micha{\l} Pilipczuk]{Institute of Informatics, University of Warsaw, Poland}
\email{michal.pilipczuk@mimuw.edu.pl}
\author{Xuding Zhu}
\address[Xuding Zhu]{Department of Mathematics, Zhejiang Normal University, China}
\email{xdzhu@zjnu.edu.cn, xudingzhu@gmail.com}

\date{\today}

\begin{abstract}
We prove that if $G$ is a sparse graph --- it belongs to a fixed class of bounded expansion~$\Cc$ --- and $d\in \N$ is fixed, then
the $d$th power of $G$ can be partitioned into cliques so that contracting each of these clique to a single vertex again yields a sparse graph.
This result has several graph-theoretic and algorithmic consequences for powers of sparse graphs, 
including bounds on their subchromatic number and efficient approximation algorithms for the chromatic number and the clique number.
\end{abstract}

\maketitle

\input{intro}

\input{prelims}

\input{clustering}

\input{subchi}

\input{chi-bounds}

\input{algorithms}

\section*{Acknowledgments}
The authors would like to thank Louis Esperet for inspiring discussions.

\bibliographystyle{amsplain}
\bibliography{subchi}

\end{document}

%% file: intro.tex
\section{Introduction}\label{sec:intro}

For a graph $G$ and $d\in \N$, the {\em{$d$th power}} of $G$, denoted $G^d$, is the graph on the same vertex set as~$G$ 
where vertices $u,v$ are considered adjacent if and only if the distance between them in $G$ is at most~$d$.
In this work we consider the following question: if we assume that a graph $G$ is sparse and $d$ is a fixed constant, what can we say about the structure of the graph $G^d$?
Clearly, $G^d$ does not need to be sparse; for instance, if $G$ is a star, then $G^2$ is complete.
However, the intuition is that even if $G^d$ becomes dense, it should retain some strong structural properties derived from the sparseness of~$G$.

We shall focus on two notions of uniform, structural sparseness in graphs: {\em{bounded expansion}} and {\em{nowhere denseness}}.
A class of graphs $\Cc$ has {\em{bounded expansion}} if for every fixed $r\in \N$, one cannot find graphs with arbitrary high edge density among depth-$r$ minors of graphs from $\Cc$.
Here, a graph $H$ is a {\em{depth-$r$ minor}} of a graph $G$ if $H$ can be obtained from a subgraph of $G$ by contracting mutually disjoint connected subgraphs of radius at most $r$.
More generally, we say that $\Cc$ is {\em{nowhere dense}} if for every $r\in \N$, not all complete graphs can be obtained as depth-$r$ minors of graphs from $\Cc$.

Thus, bounded expansion and nowhere denseness can be regarded as sparseness that persists even after local modifications, which are modelled by taking bounded depth minors.
Following the introduction of these concepts by the first two authors~\cite{NesetrilM08a,NesetrilM08b,NesetrilM11}, the last decade has witnessed a rapid development of various structural techniques centered around them, 
which proved to be useful both in the combinatorial analysis and in the algorithm design on sparse graphs. 
We invite the reader to various sources, e.g.~\cite{sparsity,notes}, for a comprehensive introduction to this area.

It has been recently realized that the developed techniques can be used not only to work with sparse graphs, but also to explain the structure in dense graphs derived from sparse ones.
For instance, Gajarsk\'y et al.~\cite{GajarskyKNMPST18} considered graph classes that can be obtained from classes of bounded expansion by means of one-dimensional first-order interpretations.
They proved that such classes can be equivalently characterized by the existence of {\em{low shrubdepth colorings}}, which can be regarded as a dense analogue of {\em{low treedepth colorings}} that characterize classes
of bounded expansion~\cite{NesetrilM06}.
This connection was earlier discovered by Kwon et al.~\cite{KwonPS20} for the operation of taking a fixed-degree power of a graph, which is a particular case of a one-dimensional first-order interpretation.
A somewhat tangential direction was recently explored by Fabia\'nski et al.~\cite{FabianskiPST19}, 
who used structural properties of powers of sparse graphs, 
inspired by connections with the model-theoretic concept of {\em{stability}}, to devise efficient parameterized algorithms for
domination and independence problems on such powers.

In all the abovementioned works, taking a fixed-degree power of a graph is a prime example of a well-behaved operation which turns a sparse graph into a dense graph that retains useful structural properties.
The goal of this work is to further explore and describe these properties.

\myparagraph{Our contribution.} We prove that if $G$ is a sparse graph and $d\in \N$ is fixed, then the graph $G^d$ admits a {\em{clustering}} $\clust$ --- a partition of the vertex set into cliques --- such that
the quotient graph $G^d/\clust$ --- the graph obtained from $G^d$ by contracting each clique of $\clust$ into a single vertex --- is also sparse. See Corollaries~\ref{thm:clustering-be} and~\ref{thm:clustering-nd} 
in Section~\ref{sec:clustering} for a formal statement.
Thus, on the conceptual level, a fixed-degree power of a sparse graph can still be regarded as a sparse graph, but with some vertices replaced by possibly large cliques.

We explore several consequences of this statement. 
On the graph-theoretic side, we show that whenever $G$ comes from a fixed class of bounded expansion $\Cc$ and $d\in \N$ is fixed, the graph $G^d$ has a bounded {\em{subchromatic number}}, which
is the least number of colors needed for a coloring of vertices so that every color class induces a disjoint
union of complete graphs. We also show that the coloring number (also called {\em{degeneracy}}) of $G^d$ is within a constant multiplicative factor from both its chromatic number and its clique number.
This gives an alternative proof of the observation of Gajarsk\'y et al.~\cite{GajarskyKNMPST18} that fixed-degree powers of classes of bounded expansion are linearly $\chi$-bounded.

On the algorithmic side, we show that if we are given a graph $H$ which is promised to be an induced subgraph of a graph $G^d$, where $d\in \N$ is fixed and $G$ belongs to a fixed class $\Cc$ of bounded expansion,
then we can constructively approximate, up to a constant factor and in $\Oh(nm)$ time, 
both the chromatic number of $H$ and the clique number of $H$.
Further, we show that in such graphs $H$ there are polynomially many inclusion-wise maximal cliques, and they all can be enumerated in polynomial time. This holds even under the
milder assumption that the class $\Cc$ is only required to be nowhere dense.
An immediate consequence of this statement is a polynomial-time algorithm for finding a clique of maximum size in graphs $H$ as above.
However, here the degree of the polynomial running time is not uniformly bounded: it depends on $d$ and the class $\Cc$ in question.

%% file: prelims.tex
\section{Preliminaries}

\myparagraph{Notation.} By $\N$ we denote the set of nonnegative integers. The vertex set and the edge set of a graph $G$ are denoted by $V(G)$ and $E(G)$, respectively.
The {\em{open neighborhood}} of a vertex $u$ consists of all the neighbors of $u$ and is denoted by $N_G(u)$. The {\em{closed neighborhood}} of $u$ is $N_G[u]=N_G(u)\cup \{u\}$.
The {\em{clique number}} of $G$ --- the maximum size of a clique in $G$ --- is denoted by $\omega(G)$.
The {\em{chromatic number}} of $G$ --- the minimum number of colors needed for a proper coloring of $G$ --- is denoted by $\chi(G)$.
Clearly, for every graph $G$ we have $\chi(G)\geq \omega(G)$.

A {\em{class of graphs}} is just a (usually infinite) set of graphs. A {\em{graph parameter}} is a function that maps graphs to nonnegative integers; examples include the clique number and the chromatic number.
For a class of graphs $\Cc$ and a graph parameter $\pi(\cdot)$, we denote $\pi(\Cc)=\sup_{G\in \Cc}\pi(G)$.
We say that $\Cc$ has {\em{bounded $\pi$}} if $\pi(\Cc)$ is finite.
For a class $\Cc$, by $\cls{\Cc}$ we denote the {\em{hereditary closure}} of $\Cc$, that is, the class comprising all induced subgraphs of graphs from $\Cc$.

For $d\in \N$ and a graph $G$, the {\em{$d$th power}} of $G$ is the graph $G^d$ on the same vertex set as $G$ such that two vertices $u,v$ are adjacent in $G^d$ if and only if the distance between them in $G$ is at most $d$.
For a class of graphs $\Cc$, we write $\Cc^d=\{G^d\colon G\in \Cc\}$.

A {\em{partition}} of a set $U$ is a family $\clust$ of non-empty, disjoint subsets of $U$ such that $\bigcup \mathcal{X}=U$.
The elements of a partition are called {\em{blocks}}.

\myparagraph{Classes of sparse graphs.} For $r\in \N$, a graph $H$ is a {\em{depth-$r$ minor}} of a graph $G$ if $H$ can be obtained from a subgraph of $G$ by contracting pairwise disjoint connected subgraphs of radius at most~$r$.
For a graph $G$, by $\grad_r(G)$ we denote the maximum {\em{edge density}} (ratio between the number of edges and vertices) among depth-$r$ minors of $G$.
Similarly, $\omega_r(G)$ is the largest clique number among depth-$r$ minors of $G$.
With these notions in place, we can define the main notions of sparsity that are considered in this work.

\begin{definition}
 For class of graphs $\Cc$, we say that
 \begin{itemize}
  \item $\Cc$ has {\em{bounded expansion}} if $\grad_r(\Cc)$ is finite for every $r\in \N$; and
  \item $\Cc$ is {\em{nowhere dense}} if $\omega_r(\Cc)$ is finite for every $r\in \N$.
 \end{itemize}
\end{definition}

Obviously, every graph class of bounded expansion is nowhere dense, but the converse implication is not true in general; see~\cite{sparsity}.

\myparagraph{Coloring numbers.} A {\em{vertex ordering}} $\sigma$ of a graph $G$ is simply a linear order on its vertex set. We write $u<_\sigma v$ to indicate that $u$ is placed before $v$ in $\sigma$.
The {\em{coloring number}} of a vertex ordering $\sigma$ of $G$ is the quantity
$$\col(G,\sigma) = 1+\max_{u\in V(G)} |\{v\colon uv\in E(G)\textrm{ and }v<_\sigma u\}|.$$
The {\em{coloring number}} of $G$, denoted $\col(G)$, is defined as the smallest coloring number among the vertex orderings of $G$.
The quantity $\col(G)-1$ is often called the {\em{degeneracy}} of $G$.

Given $G$ and a vertex ordering $\sigma$, it is easy to compute a proper coloring of $G$ with $\col(G,\sigma)$ colors: 
iterate through the vertices in the order of $\sigma$, and assign to each vertex $u$ a color that is not present among the neighbors of $u$ that are placed before in~$\sigma$.
Further, it is well-known that a vertex ordering with the optimum coloring number can be computed in linear time by iteratively extracting from the graph a vertex with the smallest degree, and putting it
in front of all the vertices extracted before. Combining these two facts yields the following.

\begin{lemma}\label{lem:greedy}
 For every graph $G$, we have $\chi(G)\leq \col(G)$. Moreover, given $G$, a proper coloring of $G$ with $\col(G)$ colors can be computed in linear time.
\end{lemma}

Here and in the sequel, when we speak about a linear or a quadratic running time, we measure it in the total input size, which is the number of vertices plus the number of edges of the input graph.

The {\em{generalized coloring numbers}} were introduced by Kierstead and Yang~\cite{KiersteadY03} to lift the idea of the coloring number to larger distances.
For $r\in \N$, a graph $G$, a vertex ordering $\sigma$ of $G$, and two vertices $u,v$ satisfying $v\leq_\sigma u$, we say that $v$ is {\em{weakly $r$-reachable}} from $u$ if there exists a path of length at most $r$ from $u$ to $v$
such that all the vertices on this path are not placed before $v$ in $\sigma$. The set of vertices that are weakly-$r$ reachable from $u$ in $\sigma$ is denoted by $\WReach_r[G,\sigma,u]$.
Then the {\em{weak $r$-coloring number}} of $\sigma$ is
$$\wcol_{r}(G,\sigma)=\max_{u\in V(G)} |\WReach_r[G,\sigma,u]|,$$
and the {\em{weak $r$-coloring number}} of $G$, denoted $\wcol_r(G)$, is the smallest $\wcol_r(G,\sigma)$ for $\sigma$ ranging over vertex orderings of $G$.
Note that $\wcol_1(G)=\col(G)$ for every graph $G$.
We remark that there are also other, related notions of generalized coloring numbers --- {\em{strong $r$-coloring number}} and {\em{$r$-admissibility}} --- but we will not use them in this work.

As shown in~\cite{NesetrilM11,Zhu09}, weak coloring numbers can be used to characterize classes of sparse graphs.

\begin{theorem}[\cite{Zhu09}]\label{thm:wcol-be}
 A class of graphs $\Cc$ has bounded expansion if and only if $\wcol_r(\Cc)$ is finite for every $r\in \N$.
\end{theorem}

\begin{theorem}[\cite{NesetrilM11}]\label{thm:wcol-nd}
 A class of graph $\Cc$ is nowhere dense if and only if for every $r\in \N$ and $\eps>0$, there exists $c\in \N$ such that $\wcol_r(G)\leq c\cdot n^\eps$ for every $n$-vertex graph $G\in \Cc$.
\end{theorem}

Note that in Theorems~\ref{thm:wcol-be} and~\ref{thm:wcol-nd}, the vertex orderings witnessing the boundedness of weak $r$-coloring numbers for different $r\in \N$ may be different.
It is therefore natural to ask whether for every graph there exists a single vertex ordering for which all the weak $r$-coloring numbers are simultaneously small.
The following result of van den Heuvel and Kierstead~\cite{vdHeuvelK19} shows that this is indeed the case.

\begin{theorem}[\cite{vdHeuvelK19}]\label{thm:universal-ordering}
 For every graph $G$ there exists a vertex ordering $\sigma^\star$ of $G$ that satisfies
 $$\wcol_r(G,\sigma^\star)\leq (2^r+1)\cdot \wcol_{2r}(G)^{4r}\qquad\textrm{for all }r\in \N.$$
\end{theorem}

%% file: clustering.tex
\section{Clustering}\label{sec:clustering}

We now proceed to the main topic of this work: clustering properties of sparse graphs.

\begin{definition}
 A {\em{clustering}} of a graph $H$ is a partition $\clust$ of the vertex set of $H$ such that every block of $\clust$ is a clique in $H$.
 The {\em{quotient graph}} $\sfrac{H}{\clust}$ is a graph whose vertices are blocks of $\clust$, where blocks $A,B\in \clust$ are adjacent if and only if there exists $u\in A$ and $v\in B$
 such that $u$ and $v$ are adjacent in~$H$.
\end{definition}

The next statement is the main result of this paper.

\newcommand{\dhalf}{\lfloor d/2\rfloor}

\begin{theorem}\label{thm:clustering}
 Let $G$ be a graph, $\sigma$ be a vertex ordering of $G$, and $d\in \N$. Then there exists a clustering $\clust$ of $G^d$ that satisfies
 $$\wcol_r(\sfrac{G^d}{\clust})\leq \wcol_{2dr}(G,\sigma)\qquad \textrm{for all }r\in \N.$$
\end{theorem}
\begin{proof}
 For a vertex $u$, let $\ell(u)$ be the vertex of $\WReach_{\dhalf}[G,\sigma,u]$ that is the smallest in the ordering~$\sigma$.
 We define a clustering $\clust$ as follows. The blocks of $\clust$ are the equivalence classes of the following equivalence relation $\sim$ on the vertex set of $G$:
 $$u\sim v \qquad \Leftrightarrow \qquad \ell(u)=\ell(v).$$
 Observe that for all $u,v$ satisfying $u\sim v$, we have
 $$\dist(u,v)\leq \dist(u,\ell(u))+\dist(\ell(u),v)=\dist(u,\ell(u))+\dist(\ell(v),v)\leq \dhalf+\dhalf\leq d,$$
 hence $\clust$ is indeed a clustering of $G^d$.
 
 We now verify the claimed upper bound on the weak coloring numbers of $\sfrac{G^d}{\clust}$.
 For a block $A$ of~$\clust$, let $\ell(A)$ denote the common $\ell(u)$ for $u$ ranging over the elements of~$A$.
 Recall that the vertex set of $\sfrac{G^d}{\clust}$ is $\clust$.
 Then let us define a vertex ordering $\tau$ of $\sfrac{G^d}{\clust}$ as follows: for any $A,B\in \clust$, we put
 $$A <_\tau B\qquad \Leftrightarrow\qquad \ell(A)<_{\sigma}\ell(B).$$
 Note here that $\ell(A)\neq \ell(B)$ whenever $A\neq B$, hence $\tau$ is indeed a linear order on $\clust$.
 
 We now claim that for every $r\in \N$, we have $\wcol_r(\sfrac{G^d}{\clust},\tau)\leq \wcol_{2dr}(G,\sigma)$. For this, it suffices to show that for every $A\in \clust$ it holds that
 \begin{equation}\label{eq:wreach-bound}
|\WReach_r[\sfrac{G^d}{\clust},\tau,A]|\leq |\WReach_{2dr}[G,\sigma,\ell(A)]|,  
 \end{equation}
 because the right hand side is bounded by $\wcol_{2dr}(G,\sigma)$.
 
 Consider any $B\in \WReach_r[\sfrac{G^d}{\clust},\tau,A]$ and let $(C_0,C_1,\ldots,C_{p})$ be any path in $\sfrac{G^d}{\clust}$ from $C_0\coloneqq A$ to $C_p\coloneqq B$ which witnesses
 that $B\in \WReach_r[\sfrac{G^d}{\clust},\tau,A]$; that is, $p\leq r$ and $C_i\geq_\tau B$ for all $i\in \{0,\ldots,p\}$.
 For $i\in \{0,1,\ldots,p\}$, let $x_i=\ell(C_i)$.
 Since for each $i\in \{0,\ldots,p-1\}$, $C_i$ and $C_{i+1}$ are adjacent in $\sfrac{G^d}{\clust}$, we can find vertices $y_i\in C_i$ and $z_i\in C_{i+1}$ that are adjacent
 in $G^d$. That is, in $G$ there is a path $Q_i$ of length at most $d$ connecting $y_i$ with $z_i$.
 Observe that by the choice of $x_i=\ell(y_i)$ and $x_{i+1}=\ell(z_i)$, each vertex of $Q_i$ is not smaller than either $x_i$ or $x_{i+1}$ in $\sigma$: the first $\dhalf$ vertices of $Q_i$
 are not smaller than $x_i$, while the last $\dhalf$ vertices are not smaller than $x_{i+1}$.
 
 Further, by the definition of mapping $\ell(\cdot)$, we can find a path $P_i$ from $y_i$ to $x_i$ such that the length of $P_i$ is at most $\dhalf$ and all
 vertices of $P_i$ are not smaller than $x_i$ in $\sigma$. Similarly, there is a path $R_i$ from $z_i$ to $x_{i+1}$ such that the length of $R_i$ is at most $\dhalf$
 and all vertices of $R_i$ are not smaller than $x_{i+1}$ in $\sigma$. Thus, by concatenating paths
 $$P_0,Q_0,R_0,P_1,Q_1,R_1,\ldots,P_{p-1},Q_{p-1},R_{p-1}$$
 in order, we obtain a walk $W$ of length at most $p(d+2\dhalf)\leq 2dr$ connecting $x_0=\ell(A)$ with $x_p=\ell(B)$.
 As we argued, each vertex traversed by $W$ is not smaller in $\sigma$ than one of the vertices $x_0,x_1,\ldots,x_p$. 
 By the choice of $\tau$, the vertex $x_p$ is the smallest in $\sigma$ among $\{x_0,x_1,\ldots,x_p\}$.
 We conclude that all the vertices traversed by $W$ are not smaller than $x_p=\ell(B)$, hence $W$ witnesses that
 $$\ell(B)\in \WReach_{2rd}[G,\sigma,\ell(A)].$$
 Since vertices $\ell(B)$ are pairwise different for different $B\in \WReach_r[\sfrac{G^d}{\clust},\tau,A]$, this implies~\eqref{eq:wreach-bound} and concludes the proof.
\end{proof}

\begin{remark}
 The partition $\clust$ provided by Theorem~\ref{thm:clustering} has the following property:
 $$\bigcap_{a\in A}\WReach_{\dhalf}[G,\sigma,a]\neq \emptyset\qquad \textrm{for every }A\in \clust.$$
 Note that this property directly implies that every block of $\clust$ is a clique in $G^d$, which means that $\clust$ is indeed a clustering of $G^d$.
\end{remark}

Having established Theorem~\ref{thm:clustering}, we can formulate the key clustering property of powers of classes of bounded expansion. 

\begin{corollary}\label{thm:clustering-be}
 Let $\Cc$ be a class of graphs with bounded expansion and $d\in \N$.
 Then for every graph $H\in \Cc^d$ there exists a clustering $\clust_H$ such that the class $\{\sfrac{H}{\clust_H}\colon H\in \Cc^d\}$ also has bounded expansion.
\end{corollary}
\begin{proof}
 By Theorems~\ref{thm:wcol-be} and~\ref{thm:universal-ordering}, for every $p\in \N$ there is a constant $c(p)$ such that for every graph $G\in \Cc$ there exists a vertex
 ordering $\sigma^\star$ satisfying $\wcol_p(G,\sigma^\star)\leq c(p)$ for all $p\in \N$.
 Let $H=G^d\in \Cc^d$, where $G\in \Cc$.
 By applying Theorem~\ref{thm:clustering} to $G$, $\sigma^\star$, and $d$, we find a clustering $\clust_H$ of $H$
 satisfying
 $$\wcol_{r}(\sfrac{H}{\clust_H})\leq \wcol_{2rd}(G,\sigma^\star)\leq c(2rd)\qquad\textrm{for all }r\in \N.$$
 By Theorem~\ref{thm:wcol-be}, this implies that the class $\{\sfrac{H}{\clust_H}\colon H\in \Cc^d\}$ has bounded expansion.
\end{proof}

By replacing the application of Theorem~\ref{thm:wcol-be} with Theorem~\ref{thm:wcol-nd}, we can prove an analogous statement for nowhere dense classes in the same way.

\begin{corollary}\label{thm:clustering-nd}
 Let $\Cc$ be a nowhere dense class of graphs and $d\in \N$.
 Then for every graph $H\in \Cc^d$ there exists a clustering $\clust_H$ such that the class $\{\sfrac{H}{\clust_H}\colon H\in \Cc^d\}$ is also nowhere dense.
\end{corollary}

We conclude this section by discussing the algorithmic aspects of Corollaries~\ref{thm:clustering-be} and~\ref{thm:clustering-nd}.
As discussed in~\cite{vdHeuvelK19}, for any fixed class of bounded expansion $\Cc$ there is a polynomial-time algorithm that, given $G\in \Cc$, computes 
a vertex ordering $\sigma^\star$ of $G$ satisfying $\wcol_p(G,\sigma^\star)\leq c(p)$ for all $p\in \N$, where $c\colon \N\to \N$ is some function that depends on $\Cc$ only.
It is easy to see that given $G\in \Cc$ together with such an ordering~$\sigma^\star$, the clustering $\clust_H$ of $H=G^d$ provided by Corollary~\ref{thm:clustering-be} can be computed in polynomial time, 
by just a straightforward implementation of the construction used in the proof. 
By combining these two facts, we conclude that Corollary~\ref{thm:clustering-be} can be made algorithmic in the following sense: for a fixed class $\Cc$ of bounded expansion and $d\in \N$, given $G\in \Cc$ we can compute
a suitable clustering of $G^d$ in polynomial time.
A similar reasoning applies to Corollary~\ref{thm:clustering-nd} as well.
However, it is not clear if a suitable clustering of $H=G^d$ can be efficiently constructed if we are given only $H$ on input, instead of~$G$.
Note here that computing the ``pre-image'' graph $G$ given $H=G^d$ on input is far from being straightforward.
We leave finding such an efficient clustering algorithm as an open problem.

%% file: subchi.tex
\section{Subchromatic index}

 The {\em{subchromatic index}} of a graph $H$, denoted $\chisub(H)$, is the least integer $c$ with the following property: the vertices of $H$ can be colored with $c$ colors so that the subgraph induced
 by each color is a disjoint union of complete graphs. Equivalently, $\chisub(H)$ is the minimum of $\chi(\sfrac{H}{\clust})$ for $\clust$ ranging over clusterings of $H$.
 
 The subchromatic index was introduced by Albertson et al.~\cite{AlbertsonJHL89}. 
 It is a notion of coloring that is tailored to graphs that are possibly dense, and hence may have a large (regular) chromatic number, but admit good clustering properties.
 As should be expected given the results of the previous section, this idea applies well to the powers of sparse graphs.
 
\begin{theorem}\label{thm:chisub-bounded}
 For every graph $G$ and $d\in \N$, we have
 $$\chisub(G^d)\leq \wcol_{2d}(G).$$
\end{theorem}
\begin{proof}
 By applying Theorem~\ref{thm:clustering} to $G$ and a vertex ordering of $G$ with the minimum possible weak $2d$-coloring number, 
 we conclude that there exists a clustering $\clust$ of $G^d$ such that $\col(\sfrac{G^d}{\clust})\leq \wcol_{2d}(G)$.
 Then we have $\chisub(G^d)\leq \chi(\sfrac{G^d}{\clust})\leq \col(\sfrac{G^d}{\clust})\leq \wcol_{2d}(G)$.
\end{proof}

By combining Theorem~\ref{thm:chisub-bounded} with Theorems~\ref{thm:wcol-be} and~\ref{thm:wcol-nd}, 
we respectively obtain the following corollaries, which express upper bounds on the subchromatic index of powers of sparse graphs.

\begin{corollary}\label{cor:subchi-be}
 For every class of graphs $\Cc$ with bounded expansion and $d\in \N$, we have 
 $$\chisub(\Cc^d)<\infty.$$
\end{corollary}

\begin{corollary}\label{cor:subchi-nd}
 For every nowhere dense class of graphs $\Cc$, $d\in \N$, and $\eps>0$, there exists a constant $c$ such that 
 $$\chisub(G^d)\leq c\cdot n^\eps\qquad\textrm{for every }n\textrm{-vertex graph }G\in \Cc.$$
\end{corollary}

We find Corollary~\ref{cor:subchi-nd} particularly interesting due to potential connections with low shrubdepth colorings in powers of nowhere dense classes, which we already briefly mentioned in Section~\ref{sec:intro}.
Without going into technical details, shrubdepth, introduced in~\cite{GanianHNOM19}, is a depth parameter suited for the treatment of dense graphs; it can be regarded as a depth-constrained variant of cliquewidth.
Kwon et al.~\cite{KwonPS20} proved that whenever a class $\Cc$ has bounded expansion and $d\in \N$, the class $\Cc^d$ has {\em{low shrubdepth colorings}}\footnote{We remark 
that Kwon et al.~\cite{KwonPS20} only discuss bounds on the cliquewidth of graphs induced by subsets of colors, but, as noted in~\cite{GajarskyKNMPST18}, their proof actually provides bounds on their shrubdepth.}.
That is, for every $p\in \N$ there exists $M(p)\in \N$ such that every graph from $\Cc^d$ can be colored with $M(p)$ colors so that every subset of $p$ colors induces a graph of bounded shrubdepth.
This result was then lifted by Gajarsk\'y et al.~\cite{GajarskyKNMPST18} to all classes of {\em{structurally bounded expansion}}, that is, images of classes of bounded expansion under 
one-dimensional first-order interpretations.
The result of Kwon et al.~\cite{KwonPS20} suggests that in powers of nowhere dense classes, one should expect low shrubdepth colorings with $n^\eps$ colors, rather than a constant,
as phrased formally in the following conjecture.

\begin{conjecture}\label{cnj:low-shrb}
 Let $\Cc$ be a nowhere dense class of graphs and $d\in \N$.
 Then for every $p\in \N$ there exists a class of graphs $\Dd_p$ of bounded shrubdepth and, for every $\eps>0$, a constant $c_{p,\eps}$ such that
 every $n$-vertex graph $H\in \Cc^d$ admits a coloring with at most $c_{p,\eps}\cdot n^\eps$ colors in which every subset of at most $p$ colors induces in $H$ a subgraph that belongs to $\Dd_p$.
\end{conjecture}

Unfortunately, the proof of Kwon et al.~\cite{KwonPS20} for powers of classes of bounded expansion does not lift to nowhere dense classes, hence Conjecture~\ref{cnj:low-shrb} remains open.
However, since the class of {\em{cluster graphs}} --- disjoint unions of complete graphs --- has bounded shrubdepth, Corollary~\ref{cor:subchi-nd} actually proves Conjecture~\ref{cnj:low-shrb} for the case $p=1$.
This may potentially point to a new, different approach to Conjecture~\ref{cnj:low-shrb}.

\medskip

Theorem~\ref{thm:chisub-bounded} provides an upper bound on the subchromatic number of the $d$th power of a graph expressed in terms of its weak $2d$-coloring number.
It is natural to ask whether for such a bound, we could rely on weak coloring numbers for smaller distances than $2d$.
We now show that this is indeed the case: the boundedness of the weak $d$-coloring number is sufficient for bounding the subchromatic number of the $d$th power.

\begin{theorem}\label{thm:chisub-bounded-2}
 For every $d\in\N$ there exists a function $f\colon \N\to \N$ such that for every graph $G$ we have
 $$\chisub(G^d)\leq f(\wcol_d(G)).$$
\end{theorem}
\begin{proof}
 Let $c=\wcol_d(G)$ and let $\sigma$ be a vertex ordering of $G$ with optimum $\wcol_d(G,\sigma)$.
 For brevity, from now on we write $\WReach_d[u]$ for $\WReach_d[G,\sigma,u]$, etc.
 Let $\lambda$ be a greedy weak $d$-coloring of $G$ that uses $c$ colors, computed over $\sigma$:
 every consecutive vertex $u$ receives a color that is different from all the colors assigned to the other vertices of $\WReach_d[u]$.
 We assume that $\lambda$ uses colors $\{1,\ldots,c\}$.
 
 For any set $X\subseteq V(G)$ and $x\in X$, the {\em{index}} of $x$ in $X$ is its index in $X$ ordered by $\sigma$, which is a number between $1$ and $|X|$.
 Consider any $u\in V(G)$ and let $x_1,\ldots,x_p$ be the set $\WReach_{\dhalf}[u]$ ordered by $\sigma$; then $p\leq c$ and $x_p=u$.
 Note that for all $1\leq i<j\leq p$, we have $x_i\in \WReach_d[x_j]$, hence vertices $x_1,\ldots,x_p$ receive pairwise different colors under $\lambda$.
 With $u$ we associate a triple of functions $\xi(u)=(\alpha_u,\beta_u,\gamma_u)$ defined as follows:
 \begin{itemize}
  \item $\alpha_u\colon \{1,\ldots,p\}\to \{1,\ldots,c\}$ is defined as $\alpha_u(i)=\lambda(x_i)$;
  \item $\beta_u\colon \{1,\ldots,p\}\to \{0,1,\ldots,\dhalf\}$ is defined as follows: $\beta_u(i)$ is the smallest $r$ such that $x_i\in \WReach_{r}[u]$;
  \item $\gamma_u\colon \{i,j\in \{1,\ldots,p\}\colon i<j\}\to \{1,\ldots,c\}$ is defined as follows: $\beta_u(i,j)$ is the index of $x_i$ in $\WReach_d[x_j]$.
 \end{itemize}
 Observe that the cardinality of the co-domain of the mapping $\xi$ is bounded by a function of $c$ and $d$. 
 Hence, it suffices to prove that $\xi$ is a subcoloring of $G^d$, that is, for every triple of functions $(\alpha,\beta,\gamma)$ as above,
 the set of those $u$ that satisfy $\xi(u)=(\alpha,\beta,\gamma)$ induces in $G^d$ a disjoint union of complete graphs.
 This is equivalent to the following claim: there are no three vertices $u,v,w$ with same color under $\xi$ that would induce a $P_3$ in $G^d$.
 
 Suppose such three vertices exist; that is,
 \begin{itemize}
  \item $\dist(u,v)\leq d$, $\dist(v,w)\leq d$, and $\dist(u,w)>d$; and
  \item $\xi(u)=\xi(v)=\xi(w)=(\alpha,\beta,\gamma)$ for some functions $\alpha,\beta,\gamma$.
 \end{itemize}
 Let $P_{uv}$ be a shortest path connecting $u$ and $v$, and let $P_{vw}$ be a shortest path connecting $v$ and~$w$.
 
 Let $x$ be the vertex on $P_{uv}$ that is the smallest in $\sigma$.
 Then $x$ splits $P_{uv}$ into two subpaths, say of length $d_1$ and $d_2$, which witness that $x\in \WReach_{d_1}[u]$ and $x\in \WReach_{d_2}[v]$, respectively.
 Note that $d_1+d_2\leq d$. Suppose for a moment that $d_1\leq d_2$; then $d_1\leq \dhalf$. 
 Since $\alpha_u=\alpha_v=\alpha$ and $\beta_u=\beta_v=\beta$, there exists a vertex $x'\in \WReach_{d_1}[v]$ such that $\lambda(x)=\lambda(x')$.
 As $x\in \WReach_{d_2}[v]$, $x'\in \WReach_{d_1}[v]$ and $d_1+d_2\leq d$, we conclude that either $x\in \WReach_d[x']$ or $x'\in \WReach_d[x]$.
 But $\lambda(x)=\lambda(x')$, hence $x=x'$.
 In particular, $x\in \WReach_{d_1}[u]\cap \WReach_{d_1}[v]\subseteq \WReach_{\dhalf}[u]\cap \WReach_{\dhalf}[v]$.
 The other case, when $d_1\geq d_2$, is symmetric, so we may conclude that in both cases, 
 $$x\in \WReach_{\dhalf}[u]\cap \WReach_{\dhalf}[v].$$
 
 Letting $y$ be the vertex on $P_{vw}$ that is the smallest in $\sigma$, we may apply the same reasoning as above to conclude that
 $$y\in \WReach_{\dhalf}[v]\cap \WReach_{\dhalf}[w].$$
 Since $\dist(u,w)>d$, we have $x\neq y$. By symmetry we may assume that $x<_\sigma y$.
 As $x,y\in \WReach_{\dhalf}[v]$, we have $x\in \WReach_d[y]$.
 Let $i<j$ be the indices of $x$ and $y$ in $\WReach_{\dhalf}[v]$, respectively.
 
 Observe that $y$ is the only vertex of color $\lambda(y)$ in $\WReach_{\dhalf}[v]$, and similarly $y$ is the only vertex of color $\lambda(y)$ in $\WReach_{\dhalf}[w]$.
 Since $\alpha_v=\alpha_w=\alpha$, we infer that the index of $y$ in $\WReach_{\dhalf}[w]$ is also $j$.
 Since $\gamma_v=\gamma$, the index of $x$ in $\WReach_d[y]$ is $\gamma(i,j)$.
 However, as $\gamma_w=\gamma$ as well, the vertex of index $i$ in $\WReach_{\dhalf}[w]$ is exactly the vertex of index $\gamma(i,j)$ in $\WReach_d[y]$ --- which is~$x$.
 Thus $x\in \WReach_{\dhalf}[w]$, which combined with $x\in \WReach_{\dhalf}[u]$ implies that $\dist(u,w)\leq d$, a contradiction.
\end{proof}

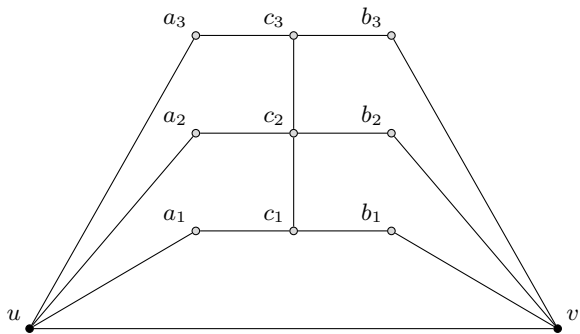
\begin{figure}[t]
\begin{tikzpicture}[scale=1.3]
\small
   \tikzstyle{nvertex}=[circle,draw=black,fill=gray!30,minimum size=0.1cm,inner sep=0pt]
   \tikzstyle{vertex}=[circle,draw=black,fill=black,minimum size=0.1cm,inner sep=0pt]

   \foreach \n/\x/\y in {a1/1/1, a2/1/2, a3/1/3, c1/2/1, c2/2/2, c3/2/3, b1/3/1, b2/3/2, b3/3/3} {
     \node[nvertex] (\n) at (\x, \y) {};
   }
   \foreach \n/\m in {a1/a_1, a2/a_2, a3/a_3}{
     \node[above left] at (\n) {$\m$};
   }
   \foreach \n/\m in {b1/b_1, b2/b_2, b3/b_3}{
     \node[above left] at (\n) {$\m$};
   }
   \foreach \n/\m in {c1/c_1, c2/c_2, c3/c_3}{
     \node[above left] at (\n) {$\m$};
   }
   
   \foreach \n/\x/\y in {u/-0.7/0, v/4.7/0} {
     \node[vertex] (\n) at (\x, \y) {};
   }
   \node[above left] at (u) {$u$};
   \node[above right] at (v) {$v$};
   \foreach \n/\m in {a1/c1, a2/c2, a3/c3, b1/c1, b2/c2, b3/c3, c1/c2, c2/c3, u/a1, u/a2, u/a3, v/b1, v/b2, v/b3, u/v} {
     \draw (\n) -- (\m);
   }
   
\end{tikzpicture}
	\caption{Gadget used in the construction of a planar graph $G$ such that $\chisub(G^2)=5$.}
	\label{fig:5gadget}
\end{figure}

Theorems~\ref{thm:chisub-bounded} and~\ref{thm:chisub-bounded-2} suggest an interesting research direction: 
finding tight bounds on the subchromatic number of powers of concrete classes of sparse graphs.
For instance, for the class of planar graphs, the following bounds on weak coloring numbers were proved in~\cite{HeuvelMQRS17}:
$$\wcol_d(\mathrm{Planar})\leq \binom{d+2}{2}(2d+1)\qquad\textrm{for all }d\in \N.$$
By combining this bound with Theorem~\ref{thm:chisub-bounded} applied for $d=2$, we conclude that
$$\chisub(\mathrm{Planar}^2)\leq 135.$$
On the other hand, the highest complementary lower bound known to us is $\chisub(\mathrm{Planar}^2)\geq 5$.
An example for this can be obtained by taking a highly branching tree with high depth, and attaching to every edge $uv$ the gadget depicted in Figure~\ref{fig:5gadget}.
We leave the easy verification that the square of the graph obtained in this way has subchromatic number equal to $5$ to the reader.
Finding tighter upper and lower bounds in the very concrete case of squares of planar graphs is an interesting open problem.

%% file: chi-bounds.tex
\section{Cliques and colorings}

Another corollary that we can derive from Theorem~\ref{thm:clustering} is that in powers of sparse graphs (and even in their induced subgraphs),
the coloring number approximates both the clique number and the chromatic number up to a constant factor.

\begin{theorem}\label{thm:omega-col}
 For every class of graphs $\Cc$ with bounded expansion and $d\in \N$, there exists a constant $c\in \N$ such that
 $$\frac{1}{c}\cdot \col(H)\leq \omega(H)\leq \chi(H)\leq \col(H)\qquad \textrm{for all }H\in \cls{\Cc^d}.$$
\end{theorem}
\begin{proof}
 The second inequality is obvious, while the third follows from Lemma~\ref{lem:greedy}. Hence, we focus on proving the first inequality.

 Let $H\in \cls{\Cc^d}$ be an induced subgraph of $G^d$, for some $G\in \Cc$.
 By Theorem~\ref{thm:clustering}, there exists a clustering $\clust$ of $G^d$ such that $\col(\sfrac{G^d}{\clust})\leq \wcol_{2d}(G)$.
 Let $\tau$ be a vertex ordering of $\sfrac{G^d}{\clust}$ with optimum coloring number 
 and let $\sigma$ be a vertex ordering of $H$ satisfying the following: whenever $A<_\tau B$ for two blocks $A,B\in \clust$, then in $\sigma$ every
 vertex of $A\cap V(H)$ is placed before every vertex of $B\cap V(H)$.
 
 Observe that for every vertex $u\in V(H)$, say belonging to $A\cap V(H)$ for a block $A\in \clust$, the neighbors of $u$ that are placed before
 $u$ in $\sigma$ either belong to $A\cap V(H)$, or to $B\cap V(H)$ for some block $B$ that is adjacent to $A$ in $\sfrac{G^d}{\clust}$ and is placed before $A$ in the ordering $\tau$.
 Therefore, we have
 $$\col(H)\leq \col(H,\sigma)\leq \col(\sfrac{G^d}{\clust},\tau)\cdot \max_{A\in \clust} |A\cap V(H)|\leq \wcol_{2d}(\Cc)\cdot \max_{A\in \clust} |A\cap V(H)|.$$
 On the other hand, for every block $A\in \clust$ the set $A\cap V(H)$ is a clique in $H$, implying
 $$\max_{A\in \clust} |A\cap V(H)|\leq \omega(H).$$
 By combining the two inequalities above we conclude that
 $$\col(H)\leq \wcol_{2d}(\Cc)\cdot \omega(H),$$
 so we can set $c=\wcol_{2d}(\Cc)$.
\end{proof}

Recall that a hereditary (i.e., closed under taking induced subgraphs) class of graphs $\Cc$ is {\em{$\chi$-bounded}} if there exists a function $f\colon \N\to \N$ such that for every $G\in \Cc$, we have
$\chi(G)\leq f(\omega(G))$. If $f(\cdot)$ can be chosen to be a linear function, then we say that $\Cc$ is {\em{linearly $\chi$-bounded}}.
Then Theorem~\ref{thm:omega-col} immediately implies that whenever $\Cc$ has bounded expansion and $d\in \N$ is fixed, the class $\cls{\Cc^d}$ is linearly $\chi$-bounded.
This result was already observed in~\cite{GajarskyKNMPST18,KwonPS20} as a consequence of the fact that fixed-degree powers of classes of bounded expansion admit low shrubdepth colorings.
The new proof provided in this paper is, however, conceptually quite different and arguably somewhat simpler.

\medskip

Given these considerations, it is natural to ask whether there is any non-trivial relation between the chromatic number and the clique number in  powers of arbitrary graphs.
It turns out that this is indeed the case, as explained next (see also the work of Chalermsook et al.~\cite{ChalermsookLN14} for similar ideas).

\begin{theorem}\label{thm:general-bounds}
 For every graph $G$ and integer $d\geq 2$, we have
 $$\chi(G^d)\leq \begin{cases} \omega(G^d)^2 & \textrm{if }d\textrm{ is even};\\ \omega(G^d)^3 & \textrm{if }d\textrm{ is odd}.\end{cases}$$
\end{theorem}
\begin{proof}
 By $\Delta(H)$ we denote the maximum degree in a graph $H$. Note that for all $p\in \N$, we have 
 $$\chi(H)\leq \Delta(H)+1,\qquad \Delta(H^{p})\leq \Delta(H)\cdot \Delta(H^{p-1}),\qquad \textrm{and}\qquad \omega(H^2)\geq \Delta(H)+1.$$
 Suppose first that $d$ is even, say $d=2k$ for some integer $k\geq 1$. Then we have
 \begin{eqnarray*}
\chi(G^{2k}) & \leq & \Delta(G^{2k})+1\leq \Delta(G^k)^2+1\\
& \leq & (\omega(G^{2k})-1)^2+1\leq \omega(G^{2k})^2.
 \end{eqnarray*}
 Now suppose that $d$ is odd, say $d=2k+1$ for some integer $k\geq 1$. Then we have
 \begin{eqnarray*}
\chi(G^{2k+1}) & \leq & \Delta(G^{2k+1})+1 \leq \Delta(G)\cdot \Delta(G^{2k})+1 \\
 & \leq & \Delta(G)\cdot \Delta(G^k)^2+1\leq \Delta(G^k)^3+1\\
 & \leq & (\omega(G^{2k})-1)^3+1\leq \omega(G^{2k+1})^3.\qedhere
 \end{eqnarray*}
\end{proof}

Note that Theorem~\ref{thm:general-bounds} cannot be extended to all induced subgraphs of degree-$d$ powers of all graphs.
Indeed, for every $d\in \N$, the hereditary closure of $\{G^d\colon G\textrm{ is a graph}\}$ comprises all graphs, as every graph is an induced subgraph of the $d$th power of its $(d-1)$-subdivision.

It would be interesting to understand to what extent the upper bounds provided by Lemma~\ref{thm:general-bounds} are tight.
The best lower bound known to us comes from the following theorem of Alon and Mohar~\cite{AlonM02}.

\begin{theorem}[\cite{AlonM02}]\label{thm:large-chromatic-power}
 For every $d,t\in \N$ there exists a constant $c\geq 0$ such that for every $r\in \N$ there exists a graph $G_r$ satisfying the following:
 $$\Delta(G_r)\leq r,\qquad \girth(G_r)\geq t,\qquad \textrm{and}\qquad \chi(G^d_r)\geq c\cdot \frac{r^d}{\log r}.$$
\end{theorem}

It is easy to observe that if $G$ has maximum degree $r$ and girth larger than $3d$, then $\omega(G^d)\leq \Oh(r^{\frac{d}{2}})$ for even $d$, and $\omega(G^d)\leq \Oh(r^{\frac{d-1}{2}})$ for odd $d$.
Then if for a fixed $d$ we combine this observation with Theorem~\ref{thm:large-chromatic-power} applied for $t=3d+1$, we can find graphs $G$ with arbitrary high $\omega(G^d)$ and satisfying
$$\chi(G^d)\geq \begin{cases} \Omega\left(\frac{\omega(G^d)^2}{\log \omega(G^d)}\right) & \textrm{if }d\textrm{ is even};\\ 
\Omega\left(\frac{\omega(G^d)^{\frac{2d}{d-1}}}{\log \omega(G^d)}\right) & \textrm{if }d\textrm{ is odd}.\end{cases}$$
This leaves a significant gap to the upper bound provided by Theorem~\ref{thm:general-bounds}, especially in the odd case.

%% file: algorithms.tex
\section{Algorithmic aspects}

Finally, we move to algorithmic applications of our graph-theoretic considerations. Theorem~\ref{thm:omega-col} together with the algorithmic tractability of the coloring number (Lemma~\ref{lem:greedy})
immediately suggests an approximation algorithm for the chromatic number of a power of a sparse graph. 

\begin{theorem}\label{thm:chromatic-apx}
 For every graph class $\Cc$ of bounded expansion and $d\in \N$, there exists a constant $c$ and a linear-time algorithm that given a graph $H\in \cls{\Cc^d}$ computes a proper coloring of $H$ with at most
 $c\cdot \chi(H)$ colors.
\end{theorem}
\begin{proof}
 By Theorem~\ref{thm:omega-col}, there exists a constant $c$ such that for every graph $H\in \cls{\Cc^d}$ we have $\col(H)\leq c\cdot \omega(H)\leq c\cdot \chi(H)$.
 Therefore, the algorithm of Lemma~\ref{lem:greedy} computes in linear time a proper coloring of $H$ with the required property.
\end{proof}

By Theorem~\ref{thm:omega-col}, whenever $H\in \cls{\Cc^d}$ for a fixed class of bounded expansion $\Cc$ and $d\in \N$, we have 
$$\frac{1}{c}\cdot \col(H)\leq \omega(H)\leq \col(H).$$
Thus, the number $\col(H)/c$ --- computable in linear-time --- can serve as a constant factor approximation of the clique number of $H$.
However, while in Theorem~\ref{thm:chromatic-apx} we have shown how to compute a proper coloring with an approximately optimum number of colors,
the above argument only provides an approximate clique number, but it is unclear how to actually construct a clique with the guaranteed size.
We next show that this is indeed possible in quadratic time.

\begin{theorem}\label{thm:clique-apx}
 For every graph class $\Cc$ of bounded expansion and $d\in \N$, there exists a constant $s\in \N$ and an $\Oh(nm)$-time algorithm that given a graph $H\in \cls{\Cc^d}$ finds a clique in $H$ of size at least
 $\frac{1}{s}\cdot \omega(H)$.
\end{theorem}

Before we proceed to the proof of Theorem~\ref{thm:clique-apx}, we need to recall some tools.
For a graph $H$, a {\em{semi-ladder}} of length $k$ in $H$ is a pair of sequences of vertices $x_1,\ldots,x_k$ and $y_1,\ldots,y_k$ such that
\begin{itemize}
 \item for every $i\in \{1,\ldots,k\}$, vertices $x_i$ and $y_i$ are not adjacent in $H$; and
 \item for all $i,j\in \{1,\ldots,k\}$ satisfying $i<j$, vertices $x_i$ and $y_j$ are adjacent in $H$.
\end{itemize}
The {\em{semi-ladder index}} of a graph $H$ is the largest length of a semi-ladder in $H$.
As usual, the {\em{semi-ladder index}} of a graph class $\Cc$ is the supremum of the semi-ladder indices of graphs from $\Cc$.

As proved by Fabia\'nski et al.~\cite{FabianskiPST19}, powers of nowhere dense classes have finite semi-ladder indices.

\begin{theorem}[\cite{FabianskiPST19}]\label{thm:semi-ladder}
 For every nowhere dense class of graphs $\Cc$ and $d\in \N$, the class $\Cc^d$ has a finite semi-ladder index.
\end{theorem}

Note that the semi-ladder index of a class is equal to the semi-ladder index of its hereditary closure,
hence Theorem~\ref{thm:semi-ladder} also implies finiteness of the semi-ladder index of $\cls{\Cc^d}$ whenever $\Cc$ is nowhere dense and $d\in \N$ is fixed.

We now give a proof of Theorem~\ref{thm:clique-apx} using Theorem~\ref{thm:semi-ladder}.

\newcommand{\AlgName}{\mathsf{ApxClique}}
\newcommand{\Myto}{\ \mathbf{to}\ }

\begin{algorithm}[t]
  \KwIn{a graph $H$} 
  \KwOut{a clique $K$ in $H$}  \Indp \BlankLine
  
  $J\leftarrow H$\\
  $K\leftarrow \emptyset$\\
  \While{$J$ is not empty}{
     $U\leftarrow$ universal vertices of $J$\\
     $J\leftarrow J-U$\\
     $K\leftarrow K\cup U$\\
     \If{$J$ is empty}{break}
     \For{$u\in V(J)$}{
        $c(u)\leftarrow \col(J[N_J[u]])$
     }
     $v\leftarrow$ any vertex of $J$ with the largest $c(v)$\\
     $J\leftarrow J[N_J(v)]$\\
     $K\leftarrow K\cup \{v\}$
  }
  \KwRet{$K$}   
\caption{Algorithm $\AlgName$}
  \label{alg:apx}
\end{algorithm}

\begin{proof}[Proof of Theorem~\ref{thm:clique-apx}]
We first present the algorithm, which is summarized using pseudo-code as procedure $\AlgName$.
The algorithm constructs a clique $K$ in an input graph $H$ by iteratively extending it, starting with $K=\emptyset$.
At the same time it maintains an induced subgraph $J$ of $H$, initially set to $H$ itself, in which the remaining vertices of the clique will be sought.
We maintain an invariant that at any point, all the vertices of $J$ are adjacent in $H$ to all the vertices of $K$.
Thus, provided all further vertices added to $K$ will be drawn from $J$, we can immediately conclude the correctness of the algorithm: the returned subset of vertices is indeed a clique.

The algorithm proceeds in rounds as follows.
First, add all universal vertices of $J$ to the clique $K$ and remove them from $J$. 
Here, a vertex is {\em{universal}} if it is adjacent to all the vertices of a graph; note that universal vertices form a clique.
Next, we iterate through all the (remaining) vertices of $J$ and we pick $v$ to be any vertex for which the coloring number of the graph $J[N_J[v]]$ is maximized.
We add $v$ to $K$, restrict $J$ to the subgraph induced by $N_J(v)$, and proceed to the next round.
The loop terminates when $J$ becomes empty. Then we output $K$ as the constructed clique.

From now on we assume that the input graph $H$ belongs to $\cls{\Cc^d}$ for some class $\Cc$ of bounded expansion and fixed $d\in \N$.
The intuition behind the algorithm is that it is an iterative self-reduction scheme that uses the computation of the coloring number as an approximate guidance for where a large clique lies in the graph.
Indeed, by Theorem~\ref{thm:omega-col}, 
any vertex of the graph $J$ that maximizes $\col(J[N_J[u]])$ participates in a clique in $J$ that is of size at least $\frac{1}{c}\cdot \omega(J)$, where $c$ is the constant provided by Theorem~\ref{thm:omega-col}
for the class $\cls{\Cc^d}$.
Thus, in each round of the algorithm we lose a constant $c$ on the approximation factor.
This is formalized in the following claim.

\begin{claim}\label{cl:apx-each-step}
 Suppose that having run procedure $\AlgName$ on some $H\in \cls{\Cc^d}$, at the end of the $i$th iteration we have obtained an induced subgraph $J_i$ of $H$ and a clique $K_i$ in $H$.
 Then
 $$|K_i|+\omega(J_i)\geq \frac{1}{c^i}\cdot \omega(H)$$
\end{claim}
\begin{proof}
We proceed by induction on $i$. The base case for $i=0$ holds vacuously. For the induction step, let $J_{i-1}$ and $K_{i-1}$ be the considered objects after $i-1$ iterations. Then by the induction hypothesis
we have
 $$|K_{i-1}|+\omega(J_{i-1})\geq \frac{1}{c^{i-1}}\cdot \omega(H).$$
Let $J'=J_{i-1}-U$, where $U$ is the set of universal vertices in $J_{i-1}$.
For every vertex $u$ of $J'$, denote $J'_u=J'[N_{J'}[u]]$.
As $J'_u$ is an induced subgraph of $H$, we have that $J'_u\in \cls{\Cc^d}$. Hence Theorem~\ref{thm:omega-col} implies that
\begin{equation}\label{eq:beaver}
\frac{1}{c}\cdot \col(J'_u)\leq \omega(J'_u)\leq \col(J'_u). 
\end{equation}
On the other hand, we have
\begin{equation}\label{eq:squirrel}
\omega(J')=\max_{u\in V(J')} \omega(J'_u),
\end{equation}
as the maximum is attained for vertices $u$ participating in maximum-size cliques in $J'$.
Let $v$ be a vertex of $J'$ that maximizes $\col(J'_v)$; say, $v$ is the vertex picked by the algorithm. By combining~\eqref{eq:beaver} and~\eqref{eq:squirrel} we infer that
$$\omega(J')=\max_{u\in V(J')} \omega(J'_u)\leq \max_{u\in V(J')} \col(J'_u)=\col(J'_v)\leq c\cdot \omega(J'_v).$$
We conclude that
\begin{eqnarray*}
|K_i|+\omega(J_i) &    = & |K_{i-1}|+|U|+1+\omega(J'_v-\{v\}) = |K_{i-1}|+|U|+\omega(J'_v)\\
                  & \geq & |K_{i-1}|+|U|+\frac{1}{c}\cdot \omega(J') \geq |K_{i-1}|+\frac{1}{c}\cdot \omega(J_{i-1}) \geq \frac{1}{c^i}\cdot \omega(H).
\end{eqnarray*}
This concludes the proof of the claim.
 \cqed\end{proof}

On the other hand, we now show that the algorithm terminates within a bounded number of rounds, as in each round it constructs a next rung of a semi-ladder.

\begin{claim}\label{cl:num-steps}
 Procedure $\AlgName$ run on any $H\in \cls{\Cc^d}$ executes at most $q$ full iterations of the loop, where $q$ is the semi-ladder index of $\Cc^d$.
\end{claim}
\begin{proof}
Suppose the algorithm executes $\ell$ full iterations.
Let $v_1,\ldots,v_\ell$ be the vertices added to the clique $K$ in consecutive iterations, in the last line of the loop.
Since each vertex $v_i$ is not universal in the graph $J_{i-1}$ (i.e., the graph at the beginning of the $i$th iteration), there is another vertex $u_i$ of $J_{i-1}$ that is not adjacent to $v_i$.
It is then easy to see that vertices $v_1,\ldots,v_\ell$ and $u_1,\ldots,u_\ell$ form a semi-ladder in $H$, implying that $\ell\leq q$.
\cqed\end{proof}

Claims~\ref{cl:apx-each-step} and~\ref{cl:num-steps} imply that the clique $K$ output by the algorithm has size at least $\frac{1}{s}\cdot \omega(H)$, where $s=c^q$. 
It remains to argue that the algorithm can be implemented so that it runs in time $\Oh(nm)$.
Observe that in each iteration of the main loop we identify and remove universal vertices --- which can be done in time $\Oh(m)$ --- and then we compute the coloring number of $\Oh(n)$ induced subgraphs.
As each such computation can be performed in time $\Oh(m)$, every iteration of the main loop takes $\Oh(nm)$ time. Finally, Claim~\ref{cl:num-steps} asserts that the algorithm terminates
after a constant number of rounds, so the total time complexity of $\Oh(nm)$ follows.
\end{proof}

Theorem~\ref{thm:clique-apx} suggests that maximum cliques in powers of classes of sparse graphs are somewhat algorithmically tractable.
We now prove that in fact, the problem of finding a maximum-size clique in powers of nowhere dense classes is polynomial-time solvable.
This comes as a corollary of the following, even stronger fact: the number of (inclusion-wise) maximal cliques is always polynomial, and they can be enumerated in polynomial time.

\begin{theorem}\label{thm:maximal-enumeration}
 For every nowhere dense class of graphs $\Cc$ and $d\in \N$, there exists $q\in \N$ such that for every $n$-vertex graph $H\in \cls{\Cc^d}$, the number of maximal cliques in 
 $H$ is at most $n^q$. Moreover, given $H\in \cls{\Cc^d}$ all maximal cliques in $H$ can be enumerated in polynomial time.
\end{theorem}
\begin{proof}
 Let $q$ be the semi-ladder index of $\Cc^d$. Then $q$ is finite by Theorem~\ref{thm:semi-ladder}.
 We prove that for every $H\in \cls{\Cc^d}$ and every maximal clique $K$ in $H$, there exists a set of vertices $A\subseteq K$ such that $|A|\leq q$ and 
 \begin{equation}\label{eq:intersection}
 K=\bigcap_{a\in A} N_H[a].
 \end{equation}
 Note this proves the theorem statement, as the number of sets $A$ of size at most $q$ is bounded by $n^q$, and in polynomial time we can enumerate all of them and for each check whether $\bigcap_{a\in A} N_H[a]$ is a maximal clique.
 
 Observe that since $K$ is inclusion-wise maximal, the set $A=K$ satisfies~\eqref{eq:intersection}.
 Therefore, we can pick $A$ to be any inclusion-wise minimal subset of $K$ satisfying~\eqref{eq:intersection}.
 It suffices to prove that then $|A|\leq q$.
 For this, observe that since $A$ is inclusion-wise minimal, for every $a\in A$ there exists $b(a)\notin K$ satisfying the following:
 $a$ and $b(a)$ are not adjacent, but $a'$ and $b(a)$ are adjacent for all $a'\in A$, $a'\neq a$.
 This implies that $A$ and $B=\{b(a)\colon a\in A\}$, with any enumerations where each $a$ and corresponding $b(a)$ receive the same index, form a semi-ladder in $H$.
 Hence $|A|\leq q$ and we are done.
\end{proof}


Theorem~\ref{thm:maximal-enumeration} suggests that maximal cliques in powers of sparse graphs are potentially an interesting object of study.
Note that the degree of the polynomial bound provided by Theorem~\ref{thm:maximal-enumeration} is dependent on the given class $\Cc$ and distance parameter $d$.
One could ask whether this could be improved to a polynomial bound whose degree is a universal constant, for instance a linear or an almost linear bound. 
Unfortunately this is not the case; the following statement provides a counterexample already in the regime of classes of bounded treedepth 
(which in particular have bounded expansion~\cite{sparsity}).

\begin{figure}[t]
\begin{tikzpicture}[scale=1.3]
\small
   \tikzstyle{nvertex}=[circle,draw=black,fill=gray!30,minimum size=0.1cm,inner sep=0pt]
   \tikzstyle{vertex}=[circle,draw=black,fill=black,minimum size=0.1cm,inner sep=0pt]

   \fill[blue!10,rounded corners=4] (-0.4,-0.4) rectangle (2.4,1.4);
   
   \foreach \n/\x/\y in {u_1/0/0, u_2/1/0, u_3/2/0, u_4/0/1, u_5/1/1, u_6/2/1} {
     \node[vertex] (\n) at (\x, \y) {};
     \node[below] at (\n) {$\n$};
   }
   
   \foreach \i in {1,2,3} {
   
   \begin{scope}[shift={(-3.5+2*\i,4)}]
    
    \fill[orange!10,rounded corners=4] (-0.4,-0.2) rectangle (1.4,1);
   
    \node[nvertex] (pI\i) at (0,0) {};
    \node[nvertex] (pIJ\i) at (0.5,0.5) {};
    \node[nvertex] (pJ\i) at (1,0) {};
    \draw (pI\i) -- (pIJ\i);
    \draw (pJ\i) -- (pIJ\i);
    
    \node[above left] at (pI\i) {$v^{\i}_I$};
    \node[above] at (pIJ\i) {$v^{\i}_{\{I,J\}}$};
    \node[above right] at (pJ\i) {$v^{\i}_J$};
    
    \foreach \n in {1,3,5} {
     \draw (pI\i) -- (u_\n);
    }
    \foreach \n in {2,4,6} {
     \draw (pJ\i) -- (u_\n);
    }
    
   \end{scope}

   }
   
\end{tikzpicture}
	\caption{Graph $G_{3,3}$. For clarity, only the vertices of $A$ and those created for the partition $\{I,J\}=\{\{1,3,5\},\{2,4,6\}\}$ are depicted.}
	\label{fig:cliques}
\end{figure}
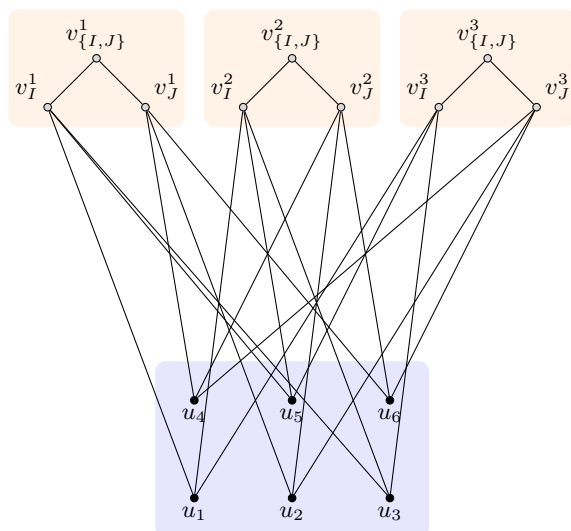

\newcommand{\Qq}{\mathcal{Q}}

\begin{theorem}\label{thm:cliques-lb}
 For all $n,d\in \N$ there exists a graph $G_{n,d}$ such that
 \begin{itemize}
  \item $G_{n,d}$ has $\frac{3}{2}\binom{2d}{d}n+2d$ vertices;
  \item the treedepth of $G_{n,d}$ is at most $2d+2$; and
  \item the graph $(G_{n,d})^2$ has at least $n^{\frac{1}{2}\binom{2d}{d}}$ different maximal cliques.
 \end{itemize}
\end{theorem}
\begin{proof}
The graph $G_{n,d}$ is constructed as follows; see Figure~\ref{fig:cliques}.
First, introduce a set of $2d$ pairwise non-adjacent vertices $A=\{u_1,u_2,\ldots,u_{2d}\}$.
Let $\Qq$ be the set of all unordered partitions of $A$ into two subsets $I$ and $J$, each of size $d$. Note that $|\Qq|=\frac{1}{2}\binom{2d}{d}$.
Next, for each $\{I,J\}\in \Qq$ and each $i\in \{1,\ldots,n\}$, introduce three vertices $v^i_I,v^i_{\{I,J\}},v^i_J$, and
make
\begin{itemize}
 \item $v^i_I$ adjacent to all the vertices of $I$;
 \item $v^i_J$ adjacent to all the vertices of $J$; and
 \item $v^i_{\{I,J\}}$ adjacent to $v^i_I$ and $v^i_J$.
\end{itemize}
This concludes the construction of $G_{n,d}$.

Observe that $G_{n,d}$ has exactly $3n|\Qq|+2d=\frac{3}{2}\binom{2d}{d}n+2d$ vertices.
To see that the treedepth of $G_{n,d}$ is at most $2d+2$, observe that
after removing the $2d$ vertices of $A$, the graph breaks into a disjoint union of paths of length $2$, and each of them has treedepth $2$.
Finally, for the last property we use the following straightforward claim, whose verification is left to the reader.

\begin{claim}\label{cl:cliques}
 For every function $f\colon \Qq\to \{1,\ldots,n\}$, the set
 $$\{v^{f(\{I,J\})}_I,v^{f(\{I,J\})}_J\,\colon\, \{I,J\}\in \Qq\}.$$
 is a maximal clique in $(G_{n,d})^2$.
\end{claim}

Thus, Claim~\ref{cl:cliques} provides $n^{|\Qq|}=n^{\frac{1}{2}\binom{2d}{d}}$ different maximal cliques in $(G_{n,d})^2$.
\end{proof}

Theorems~\ref{thm:maximal-enumeration} and~\ref{thm:cliques-lb} highlight another interesting open problem.
The running time of the algorithm of Theorem~\ref{thm:maximal-enumeration} is polynomial, but the degree of this polynomial depends on $\Cc$ and $d$.
As witnessed by Theorem~\ref{thm:cliques-lb}, this is unavoidable even for squares of classes of bounded treedepth, provided we insist on listing all maximal cliques.
However, it may be that such an exhaustive enumeration is not necessary for finding a clique of maximum cardinality.
Hence the question: may it be true that for every fixed class of bounded expansion $\Cc$ and $d\in \N$, the problem of finding a maximum-size clique in a graph from $\Cc^d$ can be solved in time
$\Oh(n^c)$, where $c$ is a univeral constant, independent of $\Cc$ and $d$? Note that Theorem~\ref{thm:clique-apx} gives such an algorithm with $c=3$ for the problem of find a clique of {\em{approximately}} maximum size,
where the approximation ratio is bounded by a constant depending on $\Cc$ and $d$.